\documentclass[sigconf,nonacm]{acmart}

\usepackage{balance}

\usepackage{graphicx}
\graphicspath{{./figures/}}

\usepackage{caption,subcaption}
\usepackage{nicefrac}
\usepackage{xspace}
\usepackage{complexity}
\usepackage{amsfonts}

\usepackage{amsthm}
\usepackage{hyperref}
\usepackage[capitalize,noabbrev,nameinlink]{cleveref}

\newtheorem{theorem}{Theorem}
\newtheorem{lemma}[theorem]{Lemma}
\newtheorem{observation}[theorem]{Observation}

\AddToHook{env/lemma/begin}{\crefalias{theorem}{lemma}}
\AddToHook{env/observation/begin}{\crefalias{theorem}{observation}}
\AddToHook{env/corollary/begin}{\crefalias{theorem}{corollary}}

\crefname{lemma}{Lemma}{Lemmas}
\crefname{observation}{Observation}{Observations}
\crefname{theorem}{Theorem}{Theorems}
\crefname{equation}{Equation}{Equations}

\title[Simplified Algorithms and Improved Bounds for Motion Planning]{Simplified Algorithms and Improved Bounds for Motion Planning}

\author{Sarita de Berg}
\authornote{Supported by VILLUM Foundation grant VIL37507 ``VIL37507 Efficient Recomputations for Changeful Problems''.}
\orcid{0000-0001-5555-966X}
\affiliation{
    \institution{
    IT University of Copenhagen}
    \city{Copenhagen}
    \country{Denmark}
    }
\email{debe@itu.dk}

\author{Joachim Gudmundsson}
\orcid{0000-0002-6778-7990}
\affiliation{
    \institution{
    University of Sydney}
    \city{Sydney}
    \country{Australia}
    }
\email{joachim.gudmundsson@sydney.edu.au}

\author{Peter Kramer}
\authornote{Partially funded by a fellowship of the German Academic Exchange Service (DAAD) and the German Research Foundation (DFG), grant 530918134.}
\orcid{0000-0001-9635-5890}
\affiliation{
    \institution{
    TU Braunschweig}
    \city{Braunschweig}
    \country{Germany}
    }
\email{kramer@ibr.cs.tu-bs.de}

\author{Christian Rieck}
\authornote{Funded by the German Research Foundation (DFG), grant 522790373.}
\orcid{0000-0003-0846-5163}
\affiliation{
    \institution{
    University of Kassel}
    \city{Kassel}
    \country{Germany}
    }
\email{christian.rieck@mathematik.uni-kassel.de}

\author{Sampson Wong}
\authornote{Partially funded by the European Union's Marie Skłodowska-Curie Actions Postdoctoral Fellowship, grant number 101146276.}
\orcid{0000-0003-3803-3804}
\affiliation{
    \institution{
    University of Copenhagen}
    \city{Copenhagen}
    \country{Denmark}
    }
\email{sampson.wong123@gmail.com}

\begin{abstract}
    We investigate coordinated motion planning~(CMP) for unit square robots in the plane under the $L_1$ metric.
    Given two arrangements of $k$ robots, we seek a movement schedule minimizing one of three objectives: the \textsc{Min-Sum} objective (sum of distances), the \textsc{Min-Makespan} objective (latest time of arrival) and the \textsc{Min-Exposure} objective (a new objective).
    
    For \textsc{Min-Sum} and \textsc{Min-Makespan}, we provide a simplified approach and improved bounds, compared to Eiben, Ganian and Kanj~(SoCG~2023).
    To this end, we partition the configuration space based on the relative ordering of the robots. This leads to a simple and efficient framework to discretize the search space. Our results extend to several settings in the plane, including
    convex robots in convex domains and vertex robots in rectangular grids.

    Finally, we apply our results to the new \textsc{Min-Exposure} objective. Given a set of polygonal regions that provide cover, the task is to find a schedule with minimal elapsed time during which at least one robot is not fully covered. We give an~$\mathcal{O}(n^4\log n)$ time algorithm that computes exposure-minimal schedules for $k=2$ robots, and an \XP algorithm for arbitrary $k$.
\end{abstract}


\newcommand{\norm}[1]{\lVert#1\rVert}
\newcommand{\abs}[1]{\lvert#1\rvert}

\newcommand{\BigO}{\mathcal{O}}

\let\xPTAS\PTAS
\renewcommand{\PTAS}{\xPTAS\xspace}
\let\xFPT\FPT
\renewcommand{\FPT}{\xFPT\xspace}
\let\xNP\NP
\renewcommand{\NP}{\xNP\xspace}
\newcommand{\paraNP}{\textup{\textsf{para-NP}}\xspace}
\let\xPSPACE\PSPACE
\renewcommand{\PSPACE}{\xPSPACE\xspace}
\let\xXP\XP
\renewcommand{\XP}{\xXP\xspace}

\newcommand{\CMP}{{\textsc{Coordinated Motion Planning}}\xspace}
\newcommand{\MMCMP}{{\textsc{Min-Makespan CMP}}\xspace}

\newcommand{\MSCMP}{{\textsc{Min-Sum CMP}}\xspace}
\newcommand{\MS}{{\textsc{Min-Sum}}\xspace}
\newcommand{\MECMP}{{\textsc{Min-Exposure CMP}}\xspace}

\newcommand{\diam}[1]{\ensuremath{d(#1)}}
\newcommand{\minSum}[1]{\ensuremath{D(#1)}}

\newcommand{\makespan}[1]{\ensuremath{\phi(#1)}}
\newcommand{\exposure}[1]{\ensuremath{\xi(#1)}}
\newcommand{\distSum}[1]{\ensuremath{\norm{#1}}}

\begin{document}

\maketitle

\section{Introduction}
\label{sec:introduction}

The growing deployment of autonomous systems in domains such as search and rescue~\cite{drew2021searchandrescue,QAMAR2023101734}, warehouse automation~\cite{0001TKDKK21,Varambally0K22}, and hazardous-material handling~\cite{hamel2001elements,trevelyan.hamel.kang2016robotics} has intensified the need for efficient coordinated multi-robot motion planning (CMP).
From the early days of computational geometry, researchers have recognized the deep algorithmic challenges inherent in planning  the coordinated motion of many robots simultaneously.
In such settings, avoiding both obstacles and collisions between robots lifts the problem into a high-dimensional configuration space whose complexity grows rapidly with the number of robots, as already demonstrated by foundational results like those of Schwartz and Sharir~\cite{schwartz.sharir1983pianomovers}.
Meeting practical performance demands, such as minimizing the sum of distances traveled or the latest time of arrival (makespan), therefore requires sophisticated geometric and algorithmic techniques.

\begin{figure}%
    \centering%
    \includegraphics[page=1]{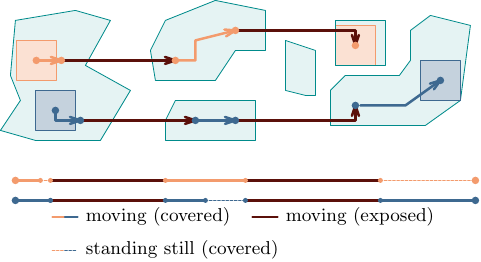}%
    \caption{
        An example instance and schedule, including a timeline, for \MECMP.
        While both square robots are in a covered region (light blue), they can move to any position within this region for free. The robots are exposed (dark red) when at least one of the robots is not fully covered.
    }
    \label{fig:exposure-figure}
\end{figure}

Two classical objectives in motion planning problems are those of minimizing the total distance traveled (\textsc{Min-Sum}) and the latest time of arrival (\textsc{Min-Makespan}); both are known to~be \NP-hard~\cite{demaine.fekete.keldenich.ea2018coordinated-motion,DemaineFKMS19,eiben.ganian.kanj2023parameterized,eiben.kanj.parsa2025motion,GeftH22}.
Some recent work focused on optimal motion for just two robots, which, surprisingly, remains a challenging task.
While results exist for various types of robots under the \textsc{Min-Sum} objective~\cite{esteban.halperin.silveira2025shortest-coordinated,kirkpatrick.liu2025minimum-length,kirkpatrick.liu2016characterizing-minimum-length}, little is known under the \textsc{Min-Makespan} objective~\cite{DemaineFKMS19}, even in the simple case of swapping two disk robots.
\begin{table*}
    \centering%
    \caption{Our results and related work on the studied variants of \CMP (CMP) for $k$ robots.}
    \begin{tabular*}{\textwidth}{@{\extracolsep{\fill}} ccclc}
        \textbf{Objective} & \textbf{Workspace}  & \textbf{Complexity} & \textbf{Running time} & \textbf{Reference}\\%
        \hline\hline%
        \MSCMP

        & rectangle in $\mathbb{R}^2$  & \FPT in $k$  &
        $\BigO^*\bigl(2^{2^{\BigO(k^2)}}\bigr)$
        & \cite{eiben.kanj.parsa2025motion}\\

        & convex pieces of $\mathbb{R}^2$ & \FPT in $k$  &
        $\BigO^*\bigl(2^{2^{\BigO(k\log k)}}\bigr)$
        & \cref{thm:fpt}\\

        & rectangular grid graph & \FPT in $k$ &
        $\BigO^*\bigl(2^{2^{\BigO(k^2 \log k)}}\bigr)$
        & \cite{eiben.ganian.kanj2023parameterized}\\

        & rectangular grid graph & \FPT in $k$ &
        $\BigO^*\bigl(2^{2^{\BigO(k \log k)}}\bigr)$
        & \cref{thm:vertex}\\
        \hline
        \MMCMP

        & convex pieces of $\mathbb{R}^2$ & \FPT in $k$&
        $\BigO^*\bigl(2^{2^{\BigO(k\log k)}}\bigr)$
        & \cref{thm:fpt}\\

        & rectangular grid graph & \FPT in $k$ &
        $\BigO^*\bigl(2^{2^{\BigO(k^2 \log k)}}\bigr)$
        & \cite{eiben.ganian.kanj2023parameterized}\\

        & rectangular grid graph & \FPT in $k$ &
        $\BigO^*\bigl(2^{2^{\BigO(k\log k)}}\bigr)$
        & \cref{thm:vertex} \\

        \hline
        \MECMP
        & $\mathbb R^2$ & \PSPACE-hard & -- & \cref{obs:min-exposure-pspace} \\
        & $\mathbb{R}^2$ & \XP in $k$  & $\BigO(n^{2k}\log n)$ & \cref{thm:min-exposure-xp}\\
        \hline
    \end{tabular*}
    \label{tab:results}
\end{table*}

In addition to these classical objectives, we consider a novel variant of continuous coordinated motion planning in which multiple agents must move to their respective target positions in the plane, while minimizing the total time during which at least one agent is \emph{exposed}.
Here, \emph{cover} is provided by a finite set of bounded regions; an agent is covered exactly if it is fully contained in the regions, as illustrated in~\cref{fig:exposure-figure}.
These bounded regions could model, for example, transmission radii in a wireless network, where agents want to stay in contact with each other as much as possible.
We call this the \MECMP problem.

Our work partly takes the perspective of parameterized complexity~\cite{DBLP:books/sp/CyganFKLMPPS15}, which has been a  focus in recent work on motion planning~\cite{eiben.ganian.kanj2023parameterized,eiben.kanj.parsa2025motion,DBLP:journals/jcss/DeligkasEGKLR26,DBLP:conf/icalp/DeligkasEGK26,DBLP:journals/talg/DeligkasEGKR26,kanj.parsa2024parameterized}.
A problem is \emph{fixed-parameter tractable}~(\FPT) in some parameter~$k$ if there exists an algorithm with runtime $\mathcal{O}^*(f(k))=f(k)\cdot n^{\mathcal{O}(1)}$, where~$f$ is a computable function and~$n$ is the input size.
If it remains \NP-hard for constant values of $k$, we speak of \paraNP-hardness.

\subsection{Our contributions}
\label{subsec:our-contributions}
Our contribution is twofold:
We give simple and exact \FPT algorithms for the \MMCMP and \MSCMP problems for $k$ unit square robots in convex domains with $n$ vertices under the $L_1$-metric.
Then, we extend our techniques to obtain exact algorithms for \MECMP.

In~\cref{sec:exact-algorithms-in-the-plane}, we describe our algorithms for \MMCMP and \MSCMP.
These generalize and improve existing results~\cite{eiben.ganian.kanj2023parameterized,eiben.kanj.parsa2025motion} in terms of runtime and complexity of description.

In~\cref{sec:min-exposure}, we formally introduce the \textsc{Min-Exposure} problem for $k$ unit square robots with cover described by polygonal domains with~$n$ vertices.
In~\cref{subsec:min-exposure-xp}, we formulate an~\XP algorithm that solves the \MECMP problem for $k$ robots. The runtime is $\mathcal{O}(n^{g(k)})$ for a computable function~$g(k)$.
A summary of related work and our results is shown in~\cref{tab:results}.

\subsection{Related work}
\label{subsec:related-work}

Research on coordinated motion planning is extensive;
we give an overview of literature most relevant to our contributions. We refer readers to existing surveys and books for broader perspectives~\cite{abrahamsen2024problemsgeobotics,AntonyshynSGM23,halperin2017robotics,halperin2017algorithmic,DBLP:books/cu/L2006,DBLP:journals/cacm/Salzman19,solovey2020,stern2019multi}.

If the number $k$ of robots moving among obstacles is part of the input, even deciding the existence of a feasible reconfiguration sequence is \PSPACE-hard in numerous models of motion~\cite{abrahamsen.buchin.buchin.ea2025reconfiguration,Hopcroft.ea.warehouseman,SoloveyH16}.
However, there exist $\mathcal{O}(n^k)$ time decision algorithms for  polygonal domains defined over $n$ vertices~\cite{ramanathan.alagar1985algorithmicmotionplanning}.

Notable recent results on \MSCMP include an $\mathcal{O}(n^4 \log n)$ time algorithm~\cite{agarwal.berg.holmgren.ea2025optimal-motion-planning} for two robots in a rectilinear polygonal environment using the $L_1$-distance, and polynomial-time approximation schemes for few robots among polygonal obstacles~\cite{argawal.halperin.sharir.steiger2024near-optimal,agarwal2026nearoptimalminsummotionplanning}.

Minimizing the overall completion time, i.e., the \MMCMP problem, is considerably more difficult.
In bounded domains, the problem is \NP-hard for just two robots~\cite{agarwal.berg.holmgren.ea2025optimal-motion-planning}, and when the number of robots is part of the input, 
it remains \NP-hard in the plane~\cite{demaine.fekete.keldenich.ea2018coordinated-motion,DemaineFKMS19}, even for constant makespan~\cite{eiben.ganian.kanj2023parameterized}.

Research on the parameterized complexity of CMP has recently received significant attention.
Eiben, Ganian and Kanj~\cite{eiben.ganian.kanj2023parameterized} first studied the parameterized complexity of \MSCMP and \MMCMP for vertex robots in rectangular integer grid graphs.
They provide \FPT algorithms parameterized by the number of robots, by following a highly technical approach that characterizes turns in optimal solutions.
In particular, they prove that all feasible instances permit \emph{canonical} optimal solutions in which each robot makes at most~$\rho(k)$ turns for some computable function~$\rho(k)>k^{k^k}$.
This is proven constructively by sequentially removing unnecessary turns in specific robots' paths, thus bounding the search space size by a function of~$\rho(k)$.
In the same paper, they show that \MSCMP is~\FPT when parameterized by the objective value, while \MMCMP is \paraNP-hard {(and thus, at least $W[1]$-hard)}.

More recently, Kanj and Parsa~\cite{kanj.parsa2024parameterized} proved the existence of an \FPT algorithm for a variant which asks for reconfiguration using the smallest number of axis-parallel moves while restricted to a bounded area.
They leverage a notion of configuration equivalence, where equivalent configurations can be reached from one another in a linear number of moves in $k$, based on the relative order of robots.
As a result, they obtain an upper bound dependent on~$k$.
This can then be used to prove the existence of an \FPT algorithm parameterized by~$k$ that runs in time~${\mathcal{O}^*( 8^{5^{k^2}\cdot 2k^2}\cdot 5^{5^{k^2}\cdot k^3}\cdot 5^{k^2})}$.
Eiben, Kanj, and Parsa~\cite{eiben.kanj.parsa2025motion} remark that this method could be adapted to obtain an \FPT algorithm for \MSCMP with rectangular robots in the plane, although no formal proof is given.

\medskip
Our results for \MS and \MMCMP in~\cref{sec:exact-algorithms-in-the-plane} can be viewed as a streamlined combination of approaches proposed in~\cite{eiben.ganian.kanj2023parameterized,eiben.kanj.parsa2025motion,kanj.parsa2024parameterized,agarwal.berg.holmgren.ea2025optimal-motion-planning}.
Our algorithm runs in time~$\mathcal{O}^*(2^{2^{\BigO(k\log k)}})$, which is a significant improvement over the running times in~\cite{eiben.ganian.kanj2023parameterized,eiben.kanj.parsa2025motion,kanj.parsa2024parameterized}. Our result applies to both \MMCMP and \MSCMP.
While our approaches are structurally similar to previous work, a key ingredient is an improved discretization of the search space, especially compared to~\cite{eiben.ganian.kanj2023parameterized}.

The novel variant of CMP, \MECMP, is closely related to path planning with weighted regions.
Weighted regions in the plane provide a framework for modeling spaces where movement incurs different costs depending on the area being traversed; the problem is then to find a minimum cost path in the plane according to the \emph{weighted}~Euclidean metric~\cite{mitchell2017geometric,mitchell.papadimitriou1991weighted-region}.
Even very restricted cases are unsolvable within the Algebraic Computation Model over the Rational Numbers~\cite{BergESS24,DECARUFEL2014724}.
There is some previous work on the $\{0,1,\infty\}$-weighted region problem~\cite{gewali.meng.mitchell.ea1988path-planning,gudmundsson_et_al:LIPIcs.WADS.2025.33}.
Note that $\{0,1\}$-weighted regions can model minimum exposure for a single robot, but not for multiple robots.

A somewhat related notion of exposure minimization was proposed by Buchin, Polishchuk, Sedov, and Voronov~\cite{DBLP:conf/compgeom/BuchinPSV20}.
They consider the problem of finding a path between two given points in a polygon that minimizes the area visible along the path.
\section{Preliminaries}
\label{sec:preliminaries}

Let the set~$\{1,\ldots, n\}$ be denoted by~$[n]$ for~$n\in\mathbb{N}^+$.
For a set~$S$ and~${i\in\mathbb{N}^+}$, the~$i$th Cartesian power $S^i$ is ${S\times \cdots \times S}$.
The concatenation of two vectors~${p\in\mathbb{R}^a}$ and~${q\in\mathbb{R}^b}$ is denoted~${(p \oplus q)\in\mathbb{R}^{a+b}}$.
For~${p\in\mathbb{R}^2}$, we denote by~$x(p)$ and~$y(p)$ its~$x$-~and~$y$-coordinates,
where $x\to \infty$ is `right' and $y\to\infty$ is `up'.
We generally use rectilinear distances, that is, the~$L_1$-norm~${\norm{p}=\abs{x(p)}+\abs{y(p)}}$.
In some cases, we use the~$L_\infty$-norm defined as~$\norm{p}_\infty=\max \{\abs{x(p)},\abs{y(p)} \}$.
Finally, the unit square corresponds to $\boxdot=\{p\in\mathbb{R}^2: \norm{p}_\infty \leq \nicefrac{1}{2}\}$.

\paragraph{Configurations.}
Given~$k$ unit square robots~${R_1,\ldots,R_k}$, we denote the center position of robot~$R_i$ by~$p_i\in\mathbb{R}^2$.
A \emph{configuration} of~$k$ robots is then a~$2k$-dimensional vector~${P=p_1\oplus \cdots \oplus p_k\in\mathbb{R}^{2k}}$.
It is \emph{feasible} if the robots have pairwise disjoint interiors, that is,~${\norm{p_i-p_j}_\infty \geq 1}$ for all~${i\neq j \in [k]}$, see \cref{fig:intro-placement-valid,fig:intro-placement-invalid}.
We denote the subset of \emph{feasible configurations} of~$k$ robots by~$\mathcal{F}_k\subseteq \mathbb{R}^{2 k}$.

It can be useful to restrict the robots' centers to subsets of the plane.
If the robot $R_i$ is restricted to a region~$S_i\subseteq \mathbb{R}^2$, we denote the \emph{feasible space restricted to $\mathcal{S}=S_1\times\cdots\times S_k\subseteq \mathbb{R}^{2k}$} by~$\mathcal{F}_k[\mathcal{S}]=\mathcal{F}_k\cap \mathcal{S}$.
If all regions coincide, we simply write~${\mathcal{F}_k[S]=\mathcal{F}_k\cap S^k}$.

\begin{figure}
    \captionsetup[subfigure]{justification=centering}%
    \begin{subfigure}[t]{0.26\columnwidth}%
        \centering%
        \includegraphics[page=1]{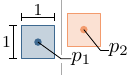}%
        \subcaption{}%
        \label{fig:intro-placement-valid}%
    \end{subfigure}%
    \begin{subfigure}[t]{0.25\columnwidth}%
        \centering%
        \includegraphics[page=2]{intro-robots}%
        \subcaption{}%
        \label{fig:intro-placement-invalid}%
    \end{subfigure}%
    \begin{subfigure}[t]{0.49\columnwidth}%
        \centering%
        \includegraphics[page=3]{intro-robots}%
        \subcaption{}%
        \label{fig:intro-trajectory}%
    \end{subfigure}
    \caption{Illustrations of (a) feasible and (b) infeasible configurations, and (c) a trajectory~$m$ with four turns.}
    \label{fig:intro-robots}%
\end{figure}

\paragraph{Trajectories.}
Robots move along polygonal \emph{trajectories} at speed up to~$1$ with respect to the~$L_1$ distance.
A~trajectory from~${a\in\mathbb{R}^2}$ to~${b\in\mathbb{R}^2}$ over the time interval~${T=[t_0, t_1]\subset \mathbb{R}}$ is a metric map\footnote{
    It holds that~$\norm{m(\alpha) - m(\beta)}\leq \abs{\alpha -\beta}$ for any~$\alpha,\beta\in T$.}
$m\colon T \rightarrow \mathbb{R}^2$, with~${m(t_0) = a}$ and~$m(t_1) = b$, whose image~$m[T]$ is a polygonal chain, see~\cref{fig:intro-trajectory}.
A~\emph{turn} in~$m$ is then any point
in time $t\in T$ at which $m$ changes the direction of velocity.
If this change is from or to zero, we speak of starting/stopping to \emph{wait}.

\paragraph{Schedules.}
A tuple of trajectories~${(m_1,\ldots, m_k)}$ is a \emph{schedule}~$M$ {for~$k$ robots }over a time interval~$T=[t_0,t_1]$.
With a slight abuse of notation, let~${M(t) =  m_1(t)\oplus \cdots \oplus m_k(t)\in\mathbb{R}^{2k}}$.
A schedule~$M$ is \emph{feasible} exactly if~$M(t)$ is feasible for every~$t\in T$, and \emph{feasible in~$\mathcal{S}\subseteq \mathbb{R}^{2k}$} if it holds that $M(t)\in \mathcal{F}_k[\mathcal{S}]$ for every~$t\in T$.

The \emph{makespan} of $M$ is defined as~$\makespan{M}=\abs{t_1-t_0}$, which is greater than or equal to the \emph{diameter}~${\diam{A,B}=\max_{i=1}^k\norm{a_i - b_i}}$, the maximum distance any robot must travel from $A$ to $B$.

The \emph{total traveled length}~$\distSum{M}=\sum_{j=1}^k\int_{t_0}^{t_1}\norm{\dot{m}_j(t)} dt$ is the sum of the trajectories' lengths.
As all robots must reach their destination, a lower bound is the sum of distances~${\minSum{A,B}=\sum_{i=1}^k\norm{a_i - b_i}}$.

\paragraph{Problem statement.}
Coordinated motion planning~(CMP) gives two feasible configurations of~$k$ robots and asks for a feasible schedule between them that optimizes some objective, such as the makespan (\MMCMP) or total traveled length (\MSCMP).
If each robot is restricted to a polygonal region, the input size is then $\Theta(m+k)$, where $m$ denotes the number of vertices.

We give all terms and definitions concerning \MECMP in~\cref{sec:min-exposure}.

\paragraph{Computational model.}
We use the Real-RAM model for our algorithms, partially because this simplifies their description significantly.
Since our approach heavily relies on solutions to rational linear programs, we remind the reader of the following result~\cite{DBLP:journals/jal/ChazelleM96}.
Similar results~\cite{DBLP:conf/focs/Vaidya89a} that account for bitwise operations can be used to show that our results transfer to such computational models.

\begin{theorem}[Chazelle and Matoušek, '96]
    Linear programs with~$n$ constraints~in $d$ dimensions can be solved in~$\BigO(d^{7d+o(d)}\,n)=\BigO(2^{\BigO(d\log d)}\, n)$ deterministic time.
    \label{thm:lp}
\end{theorem}

\section{Exact algorithms in convex regions}
\label{sec:exact-algorithms-in-the-plane}

In this section, we will introduce a unified algorithmic approach to compute optimal schedules for \MSCMP and \MMCMP restricted to convex domains.

We first characterize \emph{piecewise linear} schedules and show that, given a feasible instance, an optimal schedule of this type always exists.
Similar notions have previously been used in work on rectilinear \MSCMP~\cite{agarwal2026nearoptimalminsummotionplanning,argawal.halperin.sharir.steiger2024near-optimal,agarwal.berg.holmgren.ea2025optimal-motion-planning}, where an analogous statement was proven for decoupled (sequential) motion of square robots in polygonal domains.
The next lemma captures this idea, but in our setting.

\begin{lemma}
    Let $\mathcal{S}\subseteq\mathbb{R}^{2k}$ be any subset of the configuration space and consider any two distinct configurations~$A,B\in \mathcal{F}_k[\mathcal{S}]$.
    If the line segment~$\overline{AB}\subset \mathbb{R}^{2k}$ is fully contained in~$\mathcal{F}_k[\mathcal{S}]$,
    then linear interpolation from $A$ to $B$ induces a feasible schedule that matches the lower bounds for \MS and \MMCMP.
    \label{lem:interpolation}
\end{lemma}
\begin{proof}
    Consider linear interpolation between~$A$ and~$B$, denoted
    \begin{align}
        f: [0,1] \to \overline{AB}, \quad& \alpha \mapsto \alpha B + (1-\alpha) A
        \notag\\
        \intertext{and let~$d={d(A,B)}>0$ denote the diameter.
        We define the schedule}
        M:[0,d]\to \mathcal{F}_k[\mathcal{S}],   \quad& t\mapsto f\left(t\cdot \nicefrac{1}{d}\right).\notag
    \end{align}
    By definition, the image of this schedule consists of feasible configurations.
    The robot~$R_i$ travels at speed~${\norm{a_i-b_i}\cdot\nicefrac{1}{d} \leq \nicefrac{d}{d} \leq 1}$.
    Hence, the schedule $M$ is valid and feasible.
    Since all robots move along straight lines, no turns are made throughout.
    Finally, the schedule matches both trivial lower bounds, as we have
    \[
        \makespan{M} = d(A,B) \quad \text{and} \quad \norm{M} = D(A,B).\qedhere
    \]
\end{proof}

We call schedules that result from interpolation \emph{linear},
and call concatenations of linear schedules \emph{piecewise linear}.

\subsection{Finding piecewise linear schedules}

To apply the idea from~\cref{lem:interpolation}, we define a partition of $\mathcal{F}_k$ into convex pieces and argue based on their structure.
We will then show that all schedules can be ``straightened out'' to become piecewise linear without increasing either objective value.

For ease of exposition, we start with the case of just \emph{two} robots.
Recall that~${\mathcal{F}_2\subset \mathbb{R}^{4}}$ is obtained by removing all infeasible configurations from~$\mathbb{R}^{4}$, and that a configuration~$P$ is infeasible exactly if~${\norm{p_1-p_2}_\infty<1}$.
Equivalently, $P$ is infeasible if, along both axes, the distance between the two robots' centers is less than one.
Hence, every feasible configuration in~$\mathcal{F}_2$ permits an axis-parallel bisection of the plane such that each robot lies fully in one of the two closed half-planes.
Such an axis-parallel bisection is illustrated alongside the configuration in~\cref{fig:intro-placement-valid}.

We can transfer this observation to the configuration space.
The configurations in which $R_1$ is to the left or right of~$R_2$ are contained in the half-spaces of~$\mathbb{R}^4$ defined, respectively, by
\begin{equation}
    x(p_1) - x(p_2)\leq -1\quad \text{and} \quad x(p_1) - x(p_2) \geq 1.
    \label{eq:halfspaces-along-x}
\end{equation}
Repeating this for the~$y$-axis, we obtain four half-spaces in total.
Each half-space certifies that the contained configurations are feasible by lower-bounding the distance between the robots along one axis.
Recall that a configuration must be contained in at least one of the four half-spaces to be feasible.
Hence, the hyperplanes obtained by replacing the inequalities in~\cref{eq:halfspaces-along-x} with equalities separate the feasible configurations from those in which $R_1$ and $R_2$ overlap.
We call these the four \emph{clearance hyperplanes} of $R_1$ and $R_2$, denoted
\begin{equation}
    x(p_1) - x(p_2) = \pm1 \quad\text{and}\quad y(p_1) - y(p_2) = \pm1.
    \label{eq:clearance-hyperplane}
\end{equation}

In~\cref{fig:partition-topology}, we show an abstract representation of the four half-spaces, their overlap, and the clearance hyperplanes.

\begin{figure}
    \centering%
    \includegraphics[page=1]{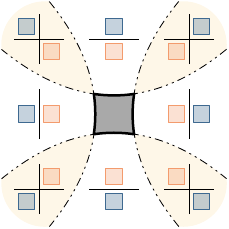}%
    \caption{
        We partition~$\mathbb{R}^4$ based on containment in the four half-spaces.
        The bounding hyperplanes are indicated by dashed arcs, yellow regions show intersections of half-spaces, and the gray cell the infeasible configurations~$\mathbb{R}^4\setminus \mathcal{F}_2$.
    }%
    \label{fig:partition-topology}
\end{figure}

To generalize this structure to $k$ robots, we define the \emph{clearance arrangement}~$\mathcal{A}_k$ as the arrangement induced by the set of hyperplanes $\mathcal{H}_k$ that contains all~${4{k\choose 2}< 2k^2}$ pairwise {clearance hyperplanes}.
That is, four for every pair of robots.
We define an \emph{open cell} of $\mathcal{A}_k$ as a connected component of~${\mathbb{R}^{2k}\setminus \left(\bigcup_{H\in \mathcal{H}_k}H\right)}$, and its closure as the corresponding \emph{closed cell}.
Every closed cell is a convex polyhedral region in~$\mathbb{R}^{2k}$ that is bounded by at most~$2k^2$ \emph{facets}, each a convex subset of a clearance hyperplane.
Recall that no open cell contains both feasible and infeasible configurations; we say a closed cell is \emph{feasible} if it is the closure of a feasible open~cell.

\begin{lemma}
    The clearance arrangement $\mathcal{A}_k$ has $2^{\BigO(k\log k)}$ cells.
    \label{lem:clearance-complexity}
\end{lemma}
\begin{proof}
    It is well-known that a $d$-dimensional arrangement of~$n$ hyperplanes has at most~$\mathcal{O}(n^d)$ cells~\cite[Corollary~28.1.1]{toth2017handbook}.
    In our case, this bound becomes~$\mathcal{O}((k^2)^{2k})=\mathcal{O}(k^{4k})=2^{\mathcal{O}(k\log k)}$.
    A more precise upper bound of~$2^{2 k\log k + \BigO(k)}$ can be obtained using exact combinatorial bounds~\cite[Theorem~28.1.1]{toth2017handbook}.
\end{proof}

\begin{observation}
    For any convex set $\mathcal{S}\subseteq\mathbb{R}^{2k}$,
    the space~$\mathcal{F}_k[\mathcal{S}]$ within every feasible closed cell of $\mathcal{A}_k$ is connected and convex.
    \label{obs:convex-connected-components}
\end{observation}
\begin{proof}
    Let $C$ be a closed feasible cell of~$\mathcal{A}_k$.
    Then $C$ is by definition convex, and so is its intersection with $\mathcal{S}$.
    Hence, any two configurations in $C\cap \mathcal{S}$ are mutually reachable by~\cref{lem:interpolation}.
\end{proof}

This gives us the insight necessary to prove the existence of piecewise linear solutions.
Intuitively, this means that we can apply~\cref{lem:interpolation} to shortcut complex paths between cells, making turns only at the shared boundary of cells in the clearance arrangement.

\begin{lemma}
    Consider $\mathcal{F}_k[\mathcal{S}]$ for a closed convex set $\mathcal{S}\subseteq \mathbb{R}^{2k}$.
    Any feasible schedule in $\mathcal{S}$ can be made into a piecewise linear schedule that visits each feasible closed cell of the clearance arrangement at most once, without increasing makespan or traveled length.
    \label{lem:piecewise-linear-exists}
\end{lemma}
\begin{proof}
    Let~$M$ be any feasible schedule in~$\mathcal{F}_k[\mathcal{S}]$ over the time interval~${T=[t_0,t_1]}$.
    Consider a feasible, closed cell~$C$ that is visited by~$M$. Let~$T_C$ denote the time interval $[t_\mathrm{in},t_\mathrm{out}]\subseteq T$ during which~$M$ is in $C$.
    Note that this does not have to mean that $M$ is in $C$ throughout all of $T_C$.
    Unless~${t_\mathrm{in}=t_\mathrm{out}}$, we now redefine~${M}$ as linear interpolation from~$M(t_\mathrm{in})$ to~$M(t_\mathrm{out})$ over~$T_C$.
    Due to~\cref{lem:interpolation} and the convexity of~$\mathcal{S}$, this is a feasible substitution in ~$\mathcal{F}_k[\mathcal{S}]$ that does not increase the makespan or traveled length.
    Since~$C$ is convex, it now follows that~${M(t)\in C}$ exactly if~${t\in T_C}$.

    It remains to argue that the substitution does not create new, separate visits to other closed cells.
    For any other closed cell~$D$, the linear interpolation either lies fully in~$D\cap C$ or meets~$D$ at most in one endpoint:
    if $D$ contains an intermediate point of the interval, all clearance hyperplanes separating $D$ and $C$ contain the full interval, and hence both endpoints are contained in $D\cap C$.
    It follows that no new, distinct visits to $D$ are created by the substitution.

    Applying the substitution to all cells iteratively yields the desired property:
    every configuration is contained in \emph{some} closed cell, so the intervals cover $T$, and each interval yields a linear schedule.
\end{proof}

\subsection{Computing optimal schedules}

We will leverage this insight to compute optimal schedules for both objectives efficiently, yielding the main result of this section.

\begin{theorem}
    Let $S_1,\ldots, S_k$ be (possibly unbounded) convex polygonal regions in $\mathbb{R}^2$ described by $m$ inequalities in total.
    Then both \MMCMP and \MSCMP where each robot $R_i$ is restricted to $S_i$
    can be solved in $\BigO\bigl(m\, 2^{2^{\BigO(k\log k)}}\bigr)$ time.

    \label{thm:fpt}
\end{theorem}

\begin{proof}
    The convex polygonal regions $S_1,\ldots, S_k$ give rise to a convex polyhedral region $\mathcal{S} = S_1 \times \ldots \times S_k \subseteq \mathbb{R}^{2k}$ that is the subset of the configuration space where the robots are contained in their respective polygons.
    Consider two configurations $A,B\in\mathcal{F}_k[\mathcal{S}]$.
    By~\cref{lem:piecewise-linear-exists}, it suffices to consider schedules that enter every closed cell of the clearance hyperplane arrangement at most once.
    From a combinatorial perspective, such a schedule corresponds to a simple path in the intersection graph of closed cells, which hence has complexity at most~$2^{\BigO(k\log k)}$ by~\cref{lem:clearance-complexity}.

    We compute the intersection graph of closed cells in $\mathcal{F}_k[\mathcal S]$ by enumerating pairs of closed feasible cells $C,D$ and checking whether they intersect within $\mathcal{S}$, that is, whether~${C\cap D\cap \mathcal S\neq \emptyset}$.
    The enumeration of cells takes~$2^{\BigO(k\log k)}$  time using a result of Edelsbrunner, O'Rourke, and Seidel~\cite{DBLP:journals/siamcomp/EdelsbrunnerOS86}.
    Recall that a cell can be expressed with $4{k \choose 2}$ linear constraints and~$\mathcal{S}$ is determined by $m$ linear constraints.
    We can check for an intersection in~${\BigO(m\,2^{\BigO(k\log k)})}$ time via linear programming due to~\cref{thm:lp}.
    Repeating this for all $2^{\BigO(k\log k)}$ cell pairs, we thus obtain the full intersection graph in
    $\BigO(m \, 2^{\BigO(k\log k)})
    =\BigO^*(2^{\BigO(k\log k)})$ time.

    To find a combinatorial path in the intersection graph that corresponds to an optimal schedule, we enumerate and evaluate all simple paths.
    For every pair of feasible closed cells $C_A$ and $C_B$ containing $A$ and $B$, respectively, we proceed as follows.
    We enumerate all simple paths from~$C_A$ to~$C_B$ in $(2^{\BigO(k\log k)})!=2^{2^{\BigO(k\log k)}}$ time using depth-first search with backtracking.
    For each path, we evaluate the achievable objective and finally output the global optimum.

    Let ${C_1,\ldots, C_\ell}$ now denote a simple path in the intersection graph as an ordered sequence of feasible closed cells, where~${C_1=C_A}$ and~${C_\ell=C_B}$.
    Within each closed cell, convexity and~\cref{lem:interpolation} provide an optimal schedule, i.e. one that matches the lower bound.
    Hence, the missing ingredient is to pick the best intermediate configurations in the intersection of successive cells, which we can model as a linear program.
    We denote these configurations by~$I_u\in C_u\cap C_{u+1}$ for~${u\in[\ell-1]}$ with $I_u=i_{u,1}\oplus\cdots\oplus i_{u,k}$, and
    the start and target~configurations by~$I_0=A$ and~$I_{\ell}=B$.

    For \textsc{Min-Makespan}, we capture the makespan as the sum of time spent in each cell, denoted $\phi_u$ for $C_u$.
    We obtain the program
    \begin{align}
        \text{minimize }&& \phi
        &= \sum_{u=1}^{\ell} \phi_u
        \label{eq:makespan-lp-intermediate-constraints}\\
        \text{subject to }&& \phi_u
        &\geq \norm{i_{u,j}-i_{u-1,j}}
        && \forall j = 1 \ldots k,\ \forall u = 1 \ldots \ell,\notag\\
        \text{and }&& I_u
        &\in \mathcal{S}
        &&\forall u=0\ldots\ell,\notag\\
        \text{and }&& I_u
        &\in C_u\cap C_{u+1}
        &&\forall u=1\ldots\ell-1.\notag\\
        \intertext{For \textsc{Min-Sum}, we instead sum the distances traveled as follows.}
        \text{minimize }&& \Psi
        &= \sum_{u=1}^{\ell}\sum_{j=1}^k
        \norm{i_{u,j}-i_{u-1,j}}
        \label{eq:sum-lp-intermediate-constraints}\\
        \text{subject to }&& I_u
        &\in \mathcal{S}
        &&\forall u=0\ldots\ell,\notag\\
        \text{and }&& I_u
        &\in C_u\cap C_{u+1}
        &&\forall u=1\ldots\ell-1.\notag
    \end{align}

    For readability, we have simplified three aspects of these programs.
    First, the $L_1$-norm can be linearized using a constant number of auxiliary variables and constraints per dimension.
    Second, since~$\mathcal S$ is determined by $m$ linear inequalities, the containment constraint requires no more than $m$ constraints per intermediate configuration.
    Finally, restricting $I_u$ to the intersection of two cells requires at most $2\cdot 4{k \choose 2}$ linear constraints.

    Both programs have~$\BigO(k\ell)$ variables and~$\BigO((m+k^2)\ell)$ constraints in total, where $\ell\leq 2^{\BigO(k\log k)}$.
    The linear program for each path can be optimized in $\BigO\bigl((m+k^2)\,(k\ell)^{\BigO(k\ell)}\bigr)=\BigO\bigl(m\,2^{2^{\BigO(k\log k)}}\bigr)$ arithmetic operations due to~\cref{thm:lp}.
    Our total running time is
    \begin{align*}
        &\mathcal{O}\Bigl(
        \underbrace{2^{2^{\BigO(k\log k)}}}_{\text{enumerating paths}}
        \:\cdot\:
        \underbrace{m\,2^{2^{\BigO(k\log k)}}}_{\text{solving linear programs}}
        \Bigr)
        =\BigO^*\bigl(2^{2^{\BigO(k \log k)}}\bigr).
    \end{align*}
    If there is no feasible schedule from $A$ to $B$, then all linear programs will be infeasible; we can simply report this.

    Otherwise, note that each linear program has a non-negative objective and hence attains its minimum if it is feasible.
    Since we create a linear program for each of the finitely many simple paths, and~\cref{lem:piecewise-linear-exists} preserves or improves the cost of any feasible schedule, the smallest of their minima is the global optimum.

    An optimal schedule can then be recovered from the best-per\-form\-ing linear program by applying linear interpolation over the duration~$d(I_{u-1},I_u)$ between successive intermediate configurations, omitting zero-duration intervals.
\end{proof}

We remark that the algorithm for~\cref{thm:fpt} is designed with the case where each robot is restricted to a convex polygon in mind, but generalizes to allowing paths to visit an arbitrary convex polyhedral subset of the configuration space.
Our running time depends on the number, but not the structure of linear inequalities.

Furthermore, the running time of our algorithm is asymptotically unaffected if instead of two specific configurations, we are given a convex polygonal region for each robot's start and target position and may choose the exact position arbitrarily within this region.

\begin{corollary}
    Let $S_i$, $A_i$, and $B_i$ for $i\in [k]$ be (possibly unbounded) convex polygonal regions in $\mathbb{R}^2$ described by $m$ inequalities in total.
    Then both \MMCMP and \MSCMP where each robot~$R_i$
    (i) starts anywhere in $A_i$,
    (ii) ends anywhere in $B_i$, and
    (iii) is restricted to $S_i$
    can be solved in $\BigO\bigl(m\, 2^{2^{\BigO(k\log k)}}\bigr)$ time.
    \label{remark:convex-start-and-target}
\end{corollary}
\begin{proof}
    Let~${\mathbb{A}=\prod_{i=1}^k A_i}$ and~$\mathbb{B}=\prod_{i=1}^k B_i$.
    By~\cref{obs:convex-connected-components}, every cell of the clearance arrangement has convex and connected intersections with $\mathbb A$ and $\mathbb B$.
    Consider two fixed, feasible closed cells~$C$ and~$D$ such that $C \cap \mathbb{A} \cap \mathcal{S}$ and $D \cap \mathbb{B} \cap \mathcal{S}$ are both non-empty.
    To apply~\cref{thm:fpt}, we adapt the linear programs from~\mbox{\cref{eq:makespan-lp-intermediate-constraints}} and~\mbox{\cref{eq:sum-lp-intermediate-constraints}}.
    We replace $I_0=A$ with~${I_0 \in \mathbb{A} \cap C \cap \mathcal{S}}$ and $I_\ell=B$ with $I_\ell \in \mathbb{B} \cap D \cap \mathcal{S}$.
    This adds at most $4k$ variables and~$\BigO(m+k^2)$ constraints: $\BigO(m)$ for $\mathbb{A},\mathbb{B}$, and $\mathcal{S}$, as well as~$\BigO(k^2)$ for~$C$ and~$D$.
    Hence, the size of our linear programs is asymptotically unchanged.
    Additionally enumerating the at most~$2^{\BigO(k\log k)}$ cell pairs of the clearance arrangement that have non-empty intersections with~$\mathbb{A}$,~$\mathbb{B}$, and~$\mathcal{S}$ does not change the dependence on $k$.
\end{proof}

\section{Minimum exposure motion planning}
\label{sec:min-exposure}
\newcommand{\coverset}{{\ensuremath{\mathcal{X}}}}
\newcommand{\cover}{{\ensuremath{X}}}

Finally, we introduce a novel problem variant, \MECMP.
Here, we are given two configurations of $k\geq 2$ robots in the plane and a set of pairwise disjoint polygonal domains~${\coverset=\{\cover_i\}}$ with a total number of $n$ vertices that provide \emph{cover} but do not restrict movement. 
A configuration~${P\in\mathcal{F}_k}$ is \emph{covered} if each robot is fully contained in the polygonal domains in $\coverset$, and \emph{exposed} otherwise.
The \emph{exposure} $\exposure{M,\coverset}$ of a schedule $M$ is then defined as the total time during which the robots are in an exposed configuration.

Finally, the \MECMP problem asks for a schedule $M$ from~$A$ to $B$ that minimizes the exposure $\exposure{M,\coverset}$, see~\cref{fig:exposure-figure}.
The total input size is ${\Theta(n+k)}$ for $\coverset$ and the two configurations.

Full containment of a robot in a polygonal domain can also be modeled as containment of the robot's center in the Minkowski subtraction of $\boxdot$ from that domain, as depicted in~\cref{fig:minkowski-subtraction}.
\begin{figure}
    \begin{subfigure}[t]{0.5\columnwidth}%
        \centering%
        \includegraphics[page=1]{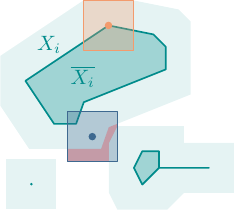}%
        \caption{}%
        \label{fig:minkowski-subtraction}
    \end{subfigure}%
    \begin{subfigure}[t]{0.5\columnwidth}%
        \centering%
        \includegraphics[page=2]{exposure-intro}%
        \caption{}%
    \end{subfigure}%
    \caption{Two configurations. (a) is exposed because of the blue robot, while (b) is shown alongside a cover~certificate.}
    \label{fig:exposure-intro}
\end{figure}
To obtain this representation, replace each polygonal domain~${\cover_i\in\coverset}$~by 
\[
    \overline{\cover_i} = \cover_i \ominus \boxdot = \{p\in\mathbb{R}^{2}\mid p+\boxdot\subseteq \cover_i\}.
\]
The input complexity does not asymptotically increase due to this operation, which can be performed in $\BigO(n\log^2 n)$ time using classical methods~\cite{DBLP:journals/dcg/KedemLPS86}.
Note that the covering geometry may degenerate as a result and, for instance, contain isolated points or line segments.
We will work with the center containment definition nonetheless, retaining degenerate components.

The \MECMP problem can be broken down into two subproblems:
Deciding whether robots can go from one covered configuration to another without being exposed and, if exposure cannot be avoided, finding a schedule that efficiently moves them into the next covered configuration.
The latter corresponds to \MMCMP as seen in~\cref{sec:exact-algorithms-in-the-plane}, and the former to \textsc{Feasibility} of reconfiguration in a polygonal domain.
Having identified this correspondence, the following observation is easily derived.
\begin{observation}
    \MECMP is \PSPACE-hard.
    \label{obs:min-exposure-pspace}
\end{observation}
\begin{proof}
    Recall that~\textsc{Feasibility} of reconfiguration in bounded domains is \PSPACE-hard, even for simple rectilinear domains~\cite{abrahamsen.buchin.buchin.ea2025reconfiguration}.
    For any two feasible configurations of $k$ robots in a simple polygonal domain, we can reduce the question of \textsc{Feasibility} to \MECMP: there exists a zero-exposure schedule exactly if reconfiguration can be achieved within the polygon.
\end{proof}

\subsection{An \XP algorithm for minimum exposure}
\label{subsec:min-exposure-xp}

We now formulate an algorithm to compute exact solutions for the \MECMP problem with a constant number of robots.
For this, we will define an auxiliary graph~$G$. This graph will help us decompose the \MECMP problem into subproblems that correspond to either \textsc{Feasibility} or \MMCMP.

Intuitively, each vertex of $G$ will correspond to a set of configurations $\Pi$ that are mutually reachable with zero exposure.
The graph will be defined such that:
(i) the shortest path distance between two vertices $\Pi$ and $\Pi'$ in $G$ matches the \MECMP between all configurations in~$\Pi$ and~$\Pi'$, and
(ii) the cost of each edge in the graph is computable via \textsc{Feasibility} or \MMCMP.
Solving \MECMP will then reduce to constructing the graph $G$ and computing a shortest path in~$G$.

\medskip
Next, we define the graph~$G$, starting with its vertices. For this, we require cover certificates.
Let $\mathcal{T}_{\coverset}$ refer to an arbitrary but fixed triangulation of the polygonal domains~$\coverset$. Note that this triangulation may contain degenerate triangles.
A cover~\emph{certificate} of a configuration $P$ defines (i) a tuple of $k$ triangles from~$\mathcal{T}_{\coverset}$, one containing the center of each robot as in~\cref{fig:exposure-intro}, and (ii)
the relative ordering of the robots in~$P$ by specifying a feasible closed cell of the clearance arrangement~$\mathcal{A}_k$ that contains the configuration.
Consider a certificate $\Pi=(\triangle_1,\ldots,\triangle_k,C)$, where $\triangle_i$ is a triangle of $\mathcal{T}_\coverset$, and~$C$ is a feasible closed cell of the arrangement.
We denote the set of configurations that share it by~$Q_\Pi\subset \mathbb{R}^{2k}$ with~${Q_\Pi=(\triangle_1\times\dots\times\triangle_k)\cap C}$; we say that~$\Pi$ is \emph{non-empty} exactly if $Q_\Pi\neq\emptyset$.
Note that $Q_\Pi$ is by definition convex and a subset of the feasible configuration space.

\begin{lemma}
    For $k$ robots and covering domains $\coverset$ with $n$ vertices,
    there are~${\BigO(n^k\,2^{\BigO(k\log k)})}$ non-empty cover certificates.
    These certificates can be computed in~${\BigO(n^k\,2^{\BigO(k\log k)})}$ time.
    
    \label{lem:cover-certificates}
\end{lemma}
\begin{proof}
    The bound on the number of certificates is clear from their structure: Every robot gets assigned exactly one of at most~$\BigO(n)$ triangles and this is combined with one of $2^{\BigO(k\log k)}$ cells.

    We first compute and fix a triangulation~$\mathcal{T}_\coverset$ of $\coverset$.
    To this end, we consider $\coverset$ as a planar, straight-line graph and compute a constrained Delaunay triangulation in $\BigO(n\log n)$ time~\cite{DBLP:conf/compgeom/Chew87}.
    Removing those triangles that are not contained in $\coverset$ gives us the desired triangulation of size~$\BigO(n)$.
    Recall that enumeration of all feasible closed cells of $\mathcal{A}_k$ takes $\BigO(2^{\BigO(k\log k)})$ time~\cite{DBLP:journals/siamcomp/EdelsbrunnerOS86}.
    We can thus enumerate all certificates in~${\BigO(n^k\cdot 2^{\BigO(k\log k)})}$ time.
    
    For each certificate $\Pi=(\triangle_1,\ldots, \triangle_k, C)$, we formulate non-empti\-ness using an LP with~$\BigO(k^2)$ rational inequality terms:
    \begin{align}
        \text{find }\quad& P \in\mathbb{R}^{2k}\label{eq:certificate-witness}\\
        \text{subject to }\quad& P \in (\triangle_1\times \dots \times \triangle_k)\cap C.\notag
    \end{align}
    We can then compute a solution in $\BigO\left(2^{\BigO(k\log k)}\,k\right)$ time due to~\cref{thm:lp}.
    Any satisfying assignment is simultaneously a witness, which we then output alongside the certificate.
\end{proof}

Enumerating and computing the cover certificates is the crucial ingredient to our approach for \MECMP.
By~\cref{lem:cover-certificates}, we can now compute the vertex set of our auxiliary graph efficiently.
In the remainder of this section, we will complete the construction of~$G$ and prove the correctness of our algorithm.

\begin{theorem}
    \MECMP is slicewise polynomial (\XP) in the number of robots.
    For $k$ robots and covering polygons with $n$ vertices, it can be solved in $\BigO(n^{2k}\,2^{2^{\BigO(k\log k)}}\,\log n)$ time.
    \label{thm:min-exposure-xp}
\end{theorem}

\begin{proof}
    Let $\coverset$ denote the set of polygonal domains that provide cover.
    We show that we can construct a graph $G$ that has as its vertices all non-empty cover certificates.
    For any two cover certificates~$\Pi,\Pi'$, the shortest path distance~$d_G(\Pi,\Pi')$ in~$G$ will be equal to the exposure incurred between $Q_\Pi$ and $Q_{\Pi'}$.

    By~\cref{lem:cover-certificates}, we can compute all non-empty cover certificates efficiently.
    We fix these as our vertex set for now and turn to the more computationally expensive task of computing edge weights.
    We proceed in two phases; first, we determine all zero-cost edges, and then all edges with non-zero cost.
    
    Let~${\Pi=(\triangle_1,\ldots, \triangle_k, C)}$ and~${\Pi'=(\triangle_1',\ldots, \triangle_k', C')}$ be two distinct certificates.
    If~${Q_{\Pi} \cap Q_{\Pi'}\neq \emptyset}$, then it follows that we can go directly from a covered configuration with certificate~$\Pi$ to a covered configuration with certificate~$\Pi'$.
    This property can be tested by combining the respective linear programs from~\cref{eq:certificate-witness} for~$\Pi$ and~$\Pi'$ and using an LP solver to determine if their constraints can be satisfied simultaneously.
    Identifying all pairwise intersections can hence be achieved in ${\BigO(n^{2k}\,2^{\BigO(k\log k)})}$ time.
    
    Using BFS, we enumerate all cover certificates in each connected component of this intersection relation and add zero-cost edges between every pair of certificates in the same component to~$G$.
    This takes at most~${\BigO(n^{2k}\cdot 2^{\BigO(k\log k)})}$ time.

    Having obtained all connected components under zero-cost movement, we now compute the minimum cost to go from one connected component to another.
    To this end, we simply enumerate all pairs of non-empty certificates $\Pi,\Pi'$.
    If they are in different connected components of $G$ under zero-cost movement, we compute a minimum-makespan schedule from $Q_\Pi$ to $Q_{\Pi'}$ in~$\BigO\bigl(k\, 2^{2^{\BigO(k\log k)}}\bigr)$
    time by~\cref{remark:convex-start-and-target}.
    This makespan determines the cost of the edge between~$\Pi$ and~$\Pi'$ in~$G$.

    Finally, we add two more vertices to the graph: the input configurations $A,B\in\mathcal{F}_k$.
    Each is connected to all remaining vertices by edges with cost equal to the necessary makespan to reach any configuration with this certificate, including an edge from $A$ to~$B$.

    Overall, construction of the graph is completed in time
    \begin{equation*}
        \BigO\bigl(
        \underbrace{n^{2k}\,2^{\BigO(k\log k)}}_{\text{vertices + zero edges}}
        +       \underbrace{n^{2k}\cdot 2^{\BigO(k\log k)}}_{\text{\# non-zero edges}}
        \cdot   \underbrace{2^{2^{\BigO(k\log k)}}}_{\text{edge costs}}
        \bigr).
    \end{equation*}
    This is asymptotically dominated by $\BigO\bigl(n^{2k}\, 2^{2^{\BigO(k\log k)}}\bigr)$.

    Finally, we use Dijkstra's algorithm to compute a shortest $AB$-path in~$G$, taking
    $\BigO\left(N\log N\right)$ time for $N=\BigO(n^{2k}\,2^{\BigO(k\log k)})$.
    The final running time is then $\BigO\bigl(n^{2k}\, 2^{2^{\BigO(k\log k)}} \log n\bigr)$, which is \XP.

    It remains to argue that our algorithm is correct. Consider any optimal solution $M$, i.e., any minimum exposure schedule between configurations~$A$ and~$B$. If $M$ never visits a covered configuration, then Theorem~\ref{thm:fpt} implies that the positive exposure edge between~$A$ and~$B$ achieves this optimum, and thus our algorithm outputs an optimal solution. Next, consider the case where $M$ visits at least one covered configuration. Define a critical configuration to be a configuration in $M$ where the robots switch from a covered to an exposed configuration, or vice versa. More formally, for a critical configuration~$F$ at time $t$, we have that all robots are covered in~$F$, but there is an $\varepsilon > 0$ such that for all $0 < \varepsilon' \leq \varepsilon$ at least one of the robots is exposed at time $t-\varepsilon'$ or $t+\varepsilon'$.
    For any covered configuration~$F \in \mathcal{F}_k$, so in particular for any critical configuration, there is a certificate $\Pi=(\triangle_1,\ldots, \triangle_k, C)$ with $F \in Q_\Pi$. Indeed, for any covered robot $R_i$ there is a triangle $\triangle_i$ of the triangulation that contains the center of $R_i$. Furthermore, as $F \in \mathcal{F}_k$ there is a feasible closed cell $C$ that contains $F$.
    Let the set of critical configurations, sorted chronologically, be $\{F_1, F_2, \ldots\}$ and let $F_0$ and $F_j$ be the first and last covered configuration in $M$. The remainder of our proof is to show that $G$ includes a path from $A$ to the certificate of~$F_0$, from the certificate of~$F_i$ to the certificate of~$F_{i+1}$, and from the certificate of $F_j$ to~$B$, where the length of each path is exactly the exposure incurred by $M$ between those configurations. It then follows that the shortest path in $G$ has length at most~$\exposure{M,\mathcal{X}}$ and our algorithm thus computes an optimal solution.

     We first consider the schedule from~$F_i$ to~$F_{i+1}$. Let $\Pi$ be the certificate of~$F_i$ and $\Pi'$ be the certificate of~$F_{i+1}$. We have two cases: Either all configurations between~$F_i$ and~$F_{i+1}$ are covered, or all configurations between~$F_i$ and~$F_{i+1}$ are exposed. In the first case, for any two consecutive certificates $\Pi_1,\Pi_2$ visited between~$F_i$ and~$F_{i+1}$, the intersection $Q_{\Pi_1} \cap Q_{\Pi_2}$ must be non-empty, and thus there is a zero-cost edge in $G$. So, there is a zero-cost path between $\Pi$ and $\Pi'$ in $G$.
     In the second case, the correctness of \cref{thm:fpt} implies that the cost of the positive exposure edge between~$\Pi$ and $\Pi'$ is equal to the optimal makespan between~$F_i$ and~$F_{i+1}$.
     Next, consider a minimum exposure schedule between~$A$ and $F_0$. Let $\Pi_{F}$ be the certificate for $F_0$. If $A$ is a covered configuration, then $A = F_0$ and we are done. Otherwise, there is no covered configuration between~$A$ and~$F_0$, and therefore, Theorem~\ref{thm:fpt} returns the correct cost for the positive exposure edge between~$A$ and $\Pi_F$. The same argument can be applied between $B$ and $F_j$.
     This completes the proof of correctness.
\end{proof}

\section{Further applications of our framework}
\subsection{Rectilinear motion in rectangular domains}
\label{subsec:rectilinear-motion}
A drawback of our approach in~\cref{sec:exact-algorithms-in-the-plane} is that, in~\cref{lem:interpolation},
we measure distance using the $L_1$-metric but do not move robots along rectilinear trajectories.
However, if each robot~$R_i$ is restricted to an \emph{arbitrary} convex region~$S_i \subseteq \mathbb{R}^2$, one can easily verify that this is unavoidable. For example, if $S_i$ is a segment on the line~${y=x}$, robot~$R_i$ can never leave its starting position using rectilinear motion.

\begin{figure*}
    \centering%
    \includegraphics[page=1]{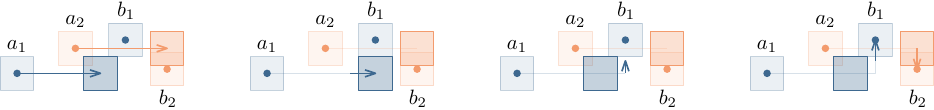}%
    \caption{Our rectilinear schedule for configurations in a common, closed feasible cell.}
    \label{fig:schedule}
\end{figure*}

We show that such issues do not exist in rectangular domains and that our methods yield rectilinear schedules in this setting.
\begin{lemma}
    Let $R_1,\ldots, R_k$ be (possibly unbounded) axis-aligned rectangles in $\mathbb{R}^2$ and let $\mathcal{R} = R_1 \times \ldots \times R_k$.
    Let $C$ be a closed feasible cell of~$\mathcal{A}_k$, restricted to $\mathcal{F}_k[\mathcal{R}]$.
    For all configurations~${A,B\in C}$ and any $d\geq d(A,B)$, there is a rectilinear schedule~$M$ in~$C$ that has $\makespan{M}=d$ and $\norm{M}=D(A,B)$.
    \label{lem:rectilinear-schedule-same-cell}
\end{lemma}
\begin{proof}
    We provide a schedule $M$ over the time interval $[0,d]$, where the trajectory $m_i$ of robot $R_i$ is defined as:
    \begin{enumerate}
        \item Move along the $x$-axis toward $x(b_i)$ at unit speed.
        \item Wait at $(x(b_i),y(a_i))$ until the time is $d-\abs{y(b_i)-y(a_i)}$.
        \item Move along the $y$-axis toward $y(b_i)$ at unit speed.
    \end{enumerate}
    Evidently, the makespan of this schedule is exactly $d$.
    Since each robot moves exclusively toward its goal, we have~$\norm{M}=D(A,B)$.
    Finally, since each robot $R_i$ moves exclusively along the axis-aligned bounding box of $a_i$ and $b_i$, its trajectory is fully contained in $R_i$.
    An example is illustrated in~\cref{fig:schedule}.
    We argue that $M$ stays in the closed feasible cell containing~$A$ and~$B$, implying feasibility.

    Consider any two robots $R_i$ and $R_j$ and consider any of their clearance hyperplanes that bounds the cell.
    Without loss of generality, let this be a hyperplane of the form~${x(p_j) - x(p_i) = \pm 1}$.
    
    While $R_i$ and~$R_j$ move along the $x$-axis in the same direction, their $x$-separation is unchanged.
    While they move in opposite horizontal directions, or if one of them is moving along the $y$-axis, their $x$-separation changes monotonically.
    Throughout $[0,d]$, the $x$-separation of the two robots thus changes monotonically from that of $A$ to that of $B$.
    If it is greater in $A$ than in $B$, for any $t\in[0,d]$, we have that
    \begin{align*}
        x(a_j) - x(a_i) \ \geq\  x(m_j(t)) - x(m_i(t)) \ \geq\  x(b_j) - x(b_i).
    \end{align*}
    Since $A$ and $B$ lie in the same closed half-space, the same holds for~$M(t)$. The same argument can be applied for hyperplanes based on $y$-separation; hence, the schedule stays in~$C$.
\end{proof}

For the case of rectangular domains,
we can now find optimal rectilinear schedules:
let $M$ be an optimal
schedule computed by~\cref{thm:fpt}.
Then we can replace $M$ during the unique visit to a feasible closed cell $C$ by a rectilinear schedule due to~\cref{lem:rectilinear-schedule-same-cell}.

\begin{corollary}
    \MSCMP and \MMCMP for $k$ robots in axis-aligned rectangles under  rectilinear motion can be solved in~$\BigO^*\bigl(2^{2^{\BigO(k\log k)}}\bigr)$ time.
\end{corollary}

\subsection{Discrete motion with vertex robots}

Our approach also works for rectangular grid graphs as studied by Eiben, Ganian, and Kanj~\cite{eiben.ganian.kanj2023parameterized}.
In this model, a \emph{grid configuration} assigns each robot a vertex in a $W\times H$ integer grid,
and it is \emph{feasible} exactly if every robot has a unique vertex.
Robots can move along edges of the grid graph in discrete time steps.
A \emph{grid trajectory} over the interval~${T=[t_0,t_1]\subset \mathbb{Z}}$ is a mapping~${m:T\to\mathbb{Z}^2}$ with~${\norm{m(t)-m(t+1)}_1\leq 1}$ for every $t_0\leq t < t_1$.
In a \emph{feasible grid schedule} $M=(m_1,\ldots, m_k)$ over $T$, the grid configuration $M(t)$ must be feasible for every $t\in T$, and robots must not move along the same edge simultaneously (they cannot swap vertices).

While the clearance arrangement~$\mathcal{A}_k$ partitions the feasible grid configurations into cells with the same fundamental properties as before, this alone is insufficient to get an \FPT algorithm.

Our first issue is that we cannot use~\cref{lem:piecewise-linear-exists} to argue that an optimal solution corresponds to a simple path in the arrangement,
as we must use rectilinear paths.
Replacing~\cref{lem:interpolation} with~\cref{lem:rectilinear-schedule-same-cell} in the proof of~\cref{lem:piecewise-linear-exists} will also not yield a comparable result for the grid setting.
By proving a weaker lemma, we can still show that there is an optimal schedule corresponding to a simple path of cells.

\begin{lemma}
    Any feasible grid schedule
    can be modified and
    partitioned into $2^{2^{\BigO(k\log k)}}$ pieces such that it remains in a specific feasible cell throughout each, without increasing makespan or traveled length.
    
    \label{lem:piecewise-rectilinear-exists}
\end{lemma}

\begin{proof}
    Let $M$ be a feasible schedule over $T=[t_0,t_1]$, and let~$C$ denote the feasible cell containing~$M(t_0)$.
    We split~$T$ into~${T_C=[t_0, t_C]}$ and $T'=[t_C+1,t_1]$,
    where $t_{C}\in T$ is the last moment in~$T$ with~${M(t_{C})\in C}$.
    We replace $M$ over the interval $[t_0, t_{C}]$ by a rectilinear schedule.
    By~\cref{lem:rectilinear-schedule-same-cell}, we can find a feasible schedule fully contained in $C$ that does not affect the objective value of~$M$.
    Afterward, we apply this procedure recursively to~$T'=[t_{C}+1, t_1]$.
    
    Because $M(t)\notin C$ for every~$t\in T'$, this recursion will yield exactly one interval for each cell~$C$ visited by $M$ in $\mathcal{A}_k$.
\end{proof}

A second key property we used in~\cref{thm:fpt} is still missing:
if a continuous schedule $M$ moves from a feasible cell~$C$ to another such cell~$C'$, there exists a point in time~${t\in\mathbb{R}}$ such that~${M(t)\in(C\cap C')}$.
With integer time steps, we might ``miss'' this~moment entirely.
For~instance, if $R_1$ and $R_2$ share an $x$-coordinate at time $0$ and move in opposite horizontal directions going to time $1$,
then no choice of ${t\in\{0,1\}}$ satisfies that ${x(m_1(t))-x(m_2(t)) = \pm 1}$, but the final configuration is contained in a different cell.

\begin{theorem}
    \MMCMP and \MSCMP for~$k$ robots in rectangular grid graphs
    can be solved in $\BigO\bigl(2^{2^{\BigO(k\log k)}}\bigr)$ time.
    
    \label{thm:vertex}
\end{theorem}
\begin{proof}
    We outline necessary changes to our proof of~\cref{thm:fpt}.
    By~\cref{lem:piecewise-rectilinear-exists}, there is an optimal schedule that follows a simple path of cells, spending one contiguous time interval in each.

    Let $C_1,\ldots, C_\ell$ be one such simple path.
    We change the linear programs in~\cref{eq:makespan-lp-intermediate-constraints,eq:sum-lp-intermediate-constraints} to account for cell transitions between discrete time steps.
    For~$u=1,\ldots,\ell-1$,
    we drop the constraint that~${I_u\in C_u\cap C_{u+1}}$ and add another configuration $I_u'$.
    We then add constraints enforcing that
    (i) $I_u\in C_u$, (ii)~${I_u'\in C_{u+1}}$, and
    (iii) $I_u$ to $I_u'$ is feasible within makespan $1$.
    
    This change to the program doubles the number of variables and multiplies the number of constraints by~$\BigO(k^2)$ to ensure the absence of pairwise collisions between $I_u$ and $I_u'$.
    
    Finally, we change from linear programs to \emph{integer} linear programs.
    Solving ILPs is \NP-hard in general~\cite{papadimitriou1981integerprogramming}, but
    for ILP with bounded dimensions, there exist analogues to~\cref{thm:lp}~\cite{FrankT87,Lenstra83,Kannan87}.
    The running time is thus unchanged compared to~\cref{thm:fpt}.
\end{proof}

\subsection{Convex robots of constant complexity}
\label{subsec:convex-robots}

If every robot is a convex polygon of constant complexity instead of a square, we can modify the clearance arrangement to account for this, without asymptotically increasing its complexity.
For a pair of convex, constant-complexity robots $R_i$ and $R_j$, their Minkowski sum can be expressed by $\BigO(1)$ linear (in)equalities.
To generalize the method used in~\cref{eq:clearance-hyperplane}, we add one clearance hyperplane per inequality. Since the number of hyperplanes does not increase asymptotically,~\cref{lem:clearance-complexity} still holds. All subsequent results from~\cref{sec:exact-algorithms-in-the-plane} hence immediately apply to convex robots.

\pagebreak
\section{Conclusion}

We presented improved bounds and simple \FPT algorithms for \MSCMP and \MMCMP with square robots in convex domains.
Our approach leverages a streamlined discretization of the configuration space for parallel motion.
Notably, unlike previous work that addresses only the \textsc{Min-Sum} objective in the plane, our results also cover the \textsc{Min-Makespan} objective.
Beyond this setting, we demonstrated that our framework extends naturally to other variants of the problem, including the \MECMP problem, rectilinear motion, and integer grids.

Our work raises follow-up questions.
Specifically, it is known that \MMCMP is \NP-hard even for a constant objective value, but its containment in \NP is, to the best of our knowledge, still open.
Our discretization could provide additional insight: if the length of relevant paths is polynomial in~$k$, membership in~\NP~follows.

\section*{Acknowledgements}
The authors would like to thank past reviewers of this manuscript for their helpful feedback and insight that led us to make significant improvements to the discretization structure in~\cref{sec:exact-algorithms-in-the-plane}.

\bibliographystyle{ACM-Reference-Format} 
\bibliography{references}

\end{document}